\documentclass{article}
\usepackage[utf8]{inputenc}

\usepackage{authblk}
\usepackage{amsmath,amssymb,amsfonts,mathrsfs,amsthm}

\newtheorem{proposition}{Proposition}

\newtheorem{example}{Example}

\title{Shared Sequencing and Latency Competition as a Noisy Contest}
\author{Akaki Mamageishvili and Jan Christoph Schlegel}
\begin{document}
\maketitle

\begin{abstract}
    We study shared sequencing for different chains from an economic angle. We introduce a minimal non-trivial model that captures cross-domain arbitrageurs' behavior and compare the performance of shared sequencing to that of separate sequencing. While shared sequencing dominates separate sequencing trivially in the sense that it makes it more likely that cross-chain arbitrage opportunities are realized, the investment and revenue comparison is more subtle: In the simple latency competition induced by First Come First Serve ordering, shared sequencing creates more wasteful latency competition compared to separate sequencing. For bidding-based sequencing, the most surprising insight is that the revenue of shared sequencing is not always higher than that of separate sequencing and depends on the transaction ordering rule applied and the arbitrage value potentially realized. 
\end{abstract}

\section{Introduction}
Cryptocurrency trading is one of the biggest use cases of blockchains that support smart contracts. Decentralized exchanges (DEXes) that run on different chains handle daily exchange volumes of billions of dollars equivalent. A lot of (potential) trading volume in DEXes is generated through arbitrage trading between different exchanges on the same chain, as well as between exchanges on different chains or between DEXes and centralized exchanges. While arbitrage opportunities on the same chain can be exploited through the atomic execution of bundles of transactions, cross-domain arbitrage is generally riskier (harder) to capture.  
We propose a game theoretic model that captures the different nature of same-domain versus cross-domain arbitrage on blockchains.
One of our main motivation for the model is recent proposals around shared sequencing for rollups. 

Rollup chains are layer-two chains built on top of Ethereum for the purpose of scaling the Ethereum main chain. They are offering lower fees and faster execution of transactions. Because of these properties, DEXes built on them already attract significant trading volume.  Rollups have a designated operator, called a sequencer, which receives transactions from users and schedules these transactions for execution. Shared sequencing schemes propose to jointly process transactions for several rollups, with the ultimate aim of improving user experience. For example, with a shared sequencer the users will be able to schedule their transaction bundles atomically even if they include transactions on different rollup chains: either all transactions will be scheduled and executed or none. A secondary aim of shared sequencing is improved economics through lower operating and maintenance costs and through collecting more of the increased arbitrage value captured by traders. More specifically, users send bundles of transactions to the shared sequencer. The sequencer distributes them to corresponding chains for execution. This gives flexibility to cross-domain arbitrageurs to get their transactions scheduled on different chains simultaneously. A proposed platform for shared sequencing is for example Espresso Systems,~\cite{espresso}.
One fundamental economic question that we address in this paper is whether shared-sequencing services really lead to improved value capture as advertised by their proponents.

As, in any trading activity, it is important how sequencers order transactions. Similarly to traditional finance exchanges, rollup sequencers usually use a first come first serve (FCFS) policy, see, e.g.,~\cite{arbitrum}.  This often generates latency competition, where parties invest into latency reduction to be competitive in arbitrage trading.\footnote{While the latency focus of competition is special to rollups, the phenomenon of transaction-ordering induced competition in blockchains is much broader and often described under the umbrella term of MEV~\cite{mev}.} 
FCFS competition may even affect the sequencer's operation. For example, Arbitrum's sequencer distributes its feed to other nodes in a fair (random) order, that is exploited by parties creating many nodes. Recently, there have been proposals to extract some fraction of MEV for rollups, see~\cite{timeboost}.
In this approach, transaction senders bid per resource unit (in the case of chains using Ethereum's EVM, per-gas) that the transaction consumes. 
Our game-theoretic model allows us to study FCFS-based transaction ordering, as well as bidding-based ordering. We are interested, in particular, in the latency investment resp.~bidding expenditure by traders in the shared sequencers versus the multiple sequencer case and derive these quantities in equilibrium.

We consider two versions of our model: the baseline version assumes that traders' expenditure on executing a transaction is irreversible and independent of whether they succeed in capturing the arbitrage. This is for example the case if traders compete in latency investment with an FCFS policy. In other cases, it might be reasonable to assume instead that traders can partially recoup their cost if their arbitrage trade fails. More precisely, if bidding for inclusion is done per resource unit of the chain (per gas), the user may save a significant fraction of their cost by adding conditional statements to the transaction that check whether the price moved away from the target and not executing the transaction if it does. This approach may save a significant portion of gas usage in case of losing the race. However, it does (slightly) increase the overall cost by adding one conditional statement, which is added independently of the outcome. If the race is lost sufficiently often, which is the case in our model with few symmetric players, such a strategy is dominant and will be applied. Therefore, we also consider a version of our model where the loser of the race for the arbitrage trade only pays a fraction of the bid.

\subsection{Related Literature}
Cross-domain extraction in the context of blockchains was first discussed in~\cite{cross_domain}. The authors argue that to extract value, arbitrageurs should execute trades on different chains sequentially and bridge assets between them. While a useful abstraction for the formalism, in a competitive setting trades should be executed in parallel, and assets may be bridged between the chains after all trades are executed. \cite{10174971} conduct an empirical analysis on cross-domain MEV. We base our modeling on this assumption of competing for the fastest execution of trades on all involved chains. A more recent survey article in the field of cross-chain MEV is~\cite{cross_domain_sok}, which describes different approaches and implementation details.
\section{Model and Results}
There are two traders, indexed by $i\in \{1,2\}$, who compete to execute an arbitrage trade.\footnote{We comment on the case of more than two traders and of more than two sequencers in the discussion section.} The trade, if successful, has a value of $v>0$ to both of the traders. There are two chains. Successful execution of the trade requires to have two transactions included in the sequencer(s) earlier than the other trader. We consider two situations: 

\begin{itemize}
    \item In the first, both transactions are processed by the same sequencer, modeling the scenario where the same sequencer serves both chains and DEXes on them. In this case, the trader can send the transaction in a single message that arrives at the sequencer with the same timestamp.
    \item In the second, the two transactions need to be sent to two independent sequencers, serving the two different chains and therefore, DEXes on them.
\end{itemize} 

The competition is modeled as follows: each trader $i$ can invest into a signal $s_i$ at a cost of $C(s_i)$, where $C$ is a strictly increasing and strictly convex function, and receives a random term $\epsilon_i$. Trader $i$ gets their transaction included earlier than $j$ if and only if $$s_i+\epsilon_i>s_j+\epsilon_j.$$

The interpretation of $s_i$ and $\epsilon_i$ can be different, depending on the choice of the cost function $C(s_i)$. We have in particular the following scenarios in mind:
\begin{enumerate}
\item The sequencer applies an FCFS transaction ordering policy. The signal $s_i$ represents an investment into latency reduction that leads to a timestamp $t(s_i+\epsilon_i)$, where $t$ is a strictly decreasing function. Here $\epsilon_i$ represents random fluctuations in the timestamps around the average.
\item The sequencer uses a combination of bidding and timestamp-based ordering such as the TimeBoost proposal in~\cite{timeboost}, where $s_i$ is the monetary bid, $\epsilon_i$ is the timestamp (assuming that the latency technology is fixed).
\item The sequencer runs a competitive block-building market similar to the one on Ethereum. The term $s_i$ represents expenditure on fees, tips, and direct or indirect payments to intermediaries (searchers, builders, etc.) and the random term $\epsilon_i$ represents frictions and other uncertain factors in the market. 
\end{enumerate}

Since our model is additive in the random terms,  it suffices to specify a distribution for the difference in signals.
Define the random variable $\Delta_{ij}=\epsilon_i-\epsilon_j$. Then, the probability of winning the single sequencer competition for $i$ is
$$F_{\Delta}(s_i-s_j),$$
where $F_{\Delta}$ is the distribution of $\Delta_{ij}$. We assume that $F_{\Delta}$ is symmetric around $0$ and uni-modal with $F_{\Delta}(0)=1/2$. That is, players are symmetric conditional on making equal investments in the signal. Therefore, the probability of winning the single sequencer competition for player $j$ is: 
$$F_{\Delta}(s_j-s_i).$$
For the two sequencer case, if the delta in the first sequencer $\Delta_{ij}^1$ and the delta in the second sequencer $\Delta_{ij}^2$ are independent and identically distributed, and each bidder produces the same signal $s_i$ for both sequencers, we have accordingly that the probability of winning the two sequencer competition for $i$ is:
$$F_{\Delta^2}(s_i-s_j)=F^2_{\Delta}(s_i-s_j),$$
where $F_{\Delta^2}$ is the distribution of the maximum of $\Delta^1_{ij}$ and $\Delta_{ij}^2$.

\subsection*{Symmetric Pure Strategy Equilibria}
Since our model is symmetric for the two traders ($F_{\Delta}$ is symmetric, and the value $v$ is the same for both),  we are interested in symmetric pure strategy Nash equilibria, which means that the two traders choose equal average signals $s_i=s_j$ which are mutual best responses to the other player's strategy. 
For the one sequencer case, a trader's best response solves

$$\max_{s_i\geq0}F_{\Delta}(s_i-s_j)v-C(s_i).$$

Assuming that the optimal signal is characterized by the first order condition and that the equilibrium is symmetric, we have candidates for equilibrium signals $s^*$ characterized by:
$$C'(s^*)=f_{\Delta}(0)v.$$

The equation is obtained by deriving the expected payoff function with respect to the signal $s_i$, setting $s_i=s_j$, and equating it to zero.
Then, the equilibrium signal is

\begin{equation}\label{equilibrium_shared}
    s^*=s^*_i=s^*_j=(C')^{-1}(f_{\Delta}(0)v).
\end{equation}

Similarly for the two sequencer case,
we have the optimization problem

\begin{equation}\label{opt_independent}
\max_{s_i\geq0}F_{\Delta}^2(s_i-s_j)v-2C(s_i).    
\end{equation}

Assuming that the optimal signal is characterized by the first order condition and that the equilibrium is symmetric, we have candidates for equilibrium signals $s^*$ characterized by:

\begin{equation}\label{condition_separate}
    C'(s^*)=f_{\Delta}(0)F_{\Delta}(0)v=f_{\Delta}(0)v/2.
\end{equation}

Then, the equilibrium signal is 

\begin{equation}\label{equilibrium_separate}
    s^*=(C')^{-1}({f(0)v}/{2}).
\end{equation}
In the following subsections, we consider different specifications of $C$ based on different sequencer designs.
\subsection{Latency Competition}\label{latency}

First, we consider a scenario where the platform uses an FCFS policy to order transactions. This is a standard choice in traditional financial exchanges as well as in most rollup sequencers. FCFS policies often lead to latency competition. That is, players invest in latency improvement. The interpretation of our model in this case is as follows: the signal can be interpreted as an average latency that can be improved by investment. The random noise gives fluctuation around the average.

A straightforward parametric specification of the cost function is to assume constant cost elasticity (which can be estimated from data), $C(s)=s^{\beta}$, where $\beta>1$ is the cost elasticity. The timestamp $t_i$ is a decreasing function (lower timestamps win the competition) of investment and the random term, $t_i(s_i+\epsilon_i)$. For example we can choose $t(s_i+\epsilon_i)=\frac{1}{s_i+\epsilon_i}$ so that $i$ wins if $t_i<t_j\Leftrightarrow s_i+\epsilon_i>s_j+\epsilon_j$. 

\begin{proposition}
For the single sequencer case with a high enough value of the trade,
$v\geq2\left(\frac{\beta}{2f(0)}\right)^{\beta}$,
there exists a unique pure strategy Nash equilibrium where both players invest
$$\left(\frac{f(0)v}{\beta}\right)^{\beta/(\beta-1)}$$ into latency reduction.
For low value of the trade, $v<2\left(\frac{\beta}{2f(0)}\right)^{\beta},$
both players in equilibrium do not invest in latency reduction.
\end{proposition}
\begin{proof}
Using the previously derived Equation~\ref{equilibrium_shared} for the optimal signaling strategy, we obtain:
$$s=\left(\frac{f(0)v}{\beta}\right)^{1/(\beta-1)}.$$
This induces an investment cost of:
$$C(s)=\left(\frac{f(0)v}{\beta}\right)^{\beta/(\beta-1)},$$
which is smaller than the expected value from the arb $v/2$ if and only if
$$\frac{v}{2}\geq\left(\frac{f(0)v}{\beta}\right)^{\beta/(\beta-1)}\Leftrightarrow v\geq 2\left(\frac{\beta}{2f(0)}\right)^{\beta}.$$ 
\end{proof}
By the same kind of proof for the case of two sequencers, we obtain:
\begin{proposition}
For the two sequencer case with a high enough value of the trade, $v\geq 8\left(\frac{\beta}{4f(0)}\right)^{\beta}$,
there exists a unique pure strategy Nash equilibrium where both players invest
$$2\left(\frac{f(0)v}{2\beta}\right)^{\beta/(\beta-1)}$$
into latency reduction.
For low value of the trade, $v<8\left(\frac{\beta}{4f(0)}\right)^{\beta},$
both players in equilibrium do not invest in latency reduction.
\end{proposition}

From the protocol perspective, investment in latency can be seen as a waste, as it is not captured by the protocol e.g. to fund protocol development and/or other users' subsidies.

A straightforward calculation shows that this waste with a shared sequencer is always larger than with two separate sequencers as long as the value of trade is large enough so that investing into latency happens in the two sequencer case.
\begin{proposition}
For a high enough value of trade $v$ (so that traders invest in latency at all), the waste with a shared sequencer is always higher than the waste with two separate sequencers. 
\end{proposition}
\begin{proof}


From the previous two propositions, the cost is higher under the shared sequencer if and only if $1\geq 2^{-\frac{1}{\beta-1}}$, which holds for any $\beta>1$. 
\end{proof}
One particularly natural case is that $\Delta$ is normally distributed with mean $0$ and standard deviation $\sigma$, in which case $f(0)=\tfrac{1}{\sqrt{2\pi}\sigma}.$
\begin{example}
For the normally distributed case, the investment into latency in the one sequencer case is $$\left(\frac{v}{\sqrt{2\pi}\beta\sigma}\right)^{\beta/(\beta-1)},$$
and for the two sequencer case is
$$2\left(\frac{v}{2\sqrt{2\pi}\beta\sigma}\right)^{\beta/(\beta-1)}.$$
Thus, investment is increasing in the value $v$, decreasing in the standard deviation of noise $\sigma$, and decreasing in cost elasticity $\beta$.
\end{example}
\subsubsection{Constant Cost Elasticity for Bidding}
The previously analyzed cost function can be given a different interpretation as the cost of bidding in a bidding-based transaction ordering scheme. This could for example be a block-based system where bidders bid for positions in the block, but random factors in the block building market perturb the ranking. Alternatively, it could be a hybrid bidding scheme, similar to TimeBoost (which we analyze in the next section) in which time advantage is bought. In the latter case, we would need to impose an upper bound on bidding needed for the liveness of the protocol. 
In the case of bidding instead of latency investment, expenditure by traders is captured by the protocol so that the comparison flips (revenue instead of waste) and a shared sequencer is preferable. 

\begin{proposition}
For constant cost elasticity and high enough value of trade $v$ (so that traders invest in latency at all), the bidding revenue with a shared sequencer is always at least as high as the bidding revenue with two separate sequencers.
\end{proposition}

The previous result assumed that there is no upper bound on the signal that bidders can produce. A natural variant is a cap $\bar{s}$ on $s$. For low values of $U$ we have an opposite result, in particular, the following holds:

\begin{proposition}
    For  $\bar{s}<\left(\frac{vf(0)}{2\beta}\right)^{1/(1-\beta)}$, the revenue with a shared sequencer is lower than the total revenue with separate sequencers.
\end{proposition}

\begin{proof}
    In the case of a low enough value of $\bar{s}$, the only pure Nash equilibrium of the single sequencer game is for both players to bid $\bar{s}$ in which case the revenue per bidder is $\bar{s}^\beta$. In the case of shared sequencers, the revenue per bidder is the minimum of $2(\frac{f(0)v}{2\beta})^{\beta/(\beta-1)}$ and $2\bar{s}^{\beta}$ which is larger than $\bar{s}^\beta$ if $\bar{s}$ is small enough.    
\end{proof}

\subsection{TimeBoost}
In this section, we consider the TimeBoost proposal,~\cite{timeboost}. In the proposal, transactions are ordered by a combination of time stamps and bids. More precisely, transactions are ordered by a score:
$$s-t,$$
where $t$ is the timestamp of the arrival of a bid at the sequencer and $0\leq s<g$ is the "time boost" which is a function of a bid submitted together with the transaction. Here the parameter $g>0$ bounds the maximal time boost to guarantee the finality of the ordering policy. Following the formula of the original TimeBoost proposal~\cite{timeboost}, we get the following fee for producing a boost of $s$:
$$C(s_i)=\frac{cs_i}{g-s_i},$$
where $c>0$ is a parameter representing an approximate normalized marginal cost. Given the cost function for the TimeBoost proposal, the marginal cost is
$$C'(s_i)=\frac{cg}{(g-s_i)^2},$$
so that our equilibrium condition becomes:
$$\frac{cg}{(g-s^*)^2}=f(0)v\Rightarrow s^*=g-\frac{\sqrt{cg}}{\sqrt{f(0)v}},$$
which for the normal distribution becomes
$$s^*=g-\frac{\sqrt{\sqrt{2\pi}cg\sigma}}{\sqrt{v}}.$$
To guarantee the existence of a pure strategy symmetric Nash equilibrium, we need to make sure that the candidate for the optimal signal through the FOC is positive, and gives positive profit (otherwise the best response to $s^*$ would be $0$ instead of $s^*$). This gives the following conditions:
\begin{proposition}
For the single sequencer case with parameters $f(0),c,g,v$ such that
$$1+\tfrac{c}{v}+\tfrac{v}{4c}\geq gf(0)$$
there exists a pure strategy Nash equilibrium where both players produce a signal of $$s^*=\begin{cases}g-\sqrt{\frac{cg}{vf(0)}},\quad &\text{ if }v>\frac{c}{gf(0)}\\0,&\text{ if }v\leq\frac{c}{gf(0)}\end{cases}$$
with cost
$$\max\{\sqrt{cgf(0)v}-c,0\}.$$
\end{proposition}
\begin{proof}
From Equation~\ref{equilibrium_shared} with $C'(s)=\frac{cg}{(g-s)^2}$ we obtain the expression for the optimal signal $s^*$. Substituting the optimal signal in the cost function $C(s)=\frac{cs}{g-s}$ gives the expression for the cost. The parameter restriction comes from the requirement that in equilibrium profit $v/2-C(s^*)$ needs to be non-negative. 
\end{proof}
By the same logic for the case two sequencers
\begin{proposition}
For the two sequencer case with parameters $f(0),c,g,v$ such that
$$1+\tfrac{c}{v}+\tfrac{v}{4c}\geq \tfrac{gf(0)}{2}$$
there exists a pure strategy Nash equilibrium where both players produce a signal $$s^*=\begin{cases}g-\sqrt{\frac{2cg}{vf(0)}},\quad &\text{ if }v>\frac{2c}{gf(0)}\\0,&\text{ if }v\leq\frac{2c}{gf(0)}\end{cases}$$
in each of the two sequencers
with cost
$$\max\{\sqrt{2cgf(0)v}-2c,0\}.$$
\end{proposition}
\begin{example}
    For the normally distributed case, the revenue in the one sequencer case is $$\sqrt{\frac{cgv}{\sqrt{2\pi}\sigma}}-c,$$
and for the two sequencer case is
$$\sqrt{\frac{2cgv}{\sqrt{2\pi}\sigma}}-2c.$$
Thus, investment is increasing in the value $v$, increasing in the standard deviation of noise $\sigma$, increasing in the time boost parameter. The effect of an increase in the parameter $c>0$ is increasing for small $c$ and decreasing for large $c$. 
\end{example}
It is interesting to compare the expenditure on signaling in both cases. We can interpret the fraction $\frac{c}{g}$ as an approximation to the marginal cost of bidding. The following proposition then states that the expenditure is larger in the two sequencer case, if the marginal cost of bidding is low, the variance in noise is small or the valuation is large. In these cases, the two sequencer case produces more revenue. Otherwise, the shared sequencer produces more revenue.
\begin{proposition}
For the time boost policy with normally distributed $\Delta$, the bidding revenue for two separate sequencers is larger than the revenue under a shared sequencer if and only $$0.068447\frac{v}{\sigma}\geq\frac{c}{g},$$ i.e. if the marginal cost of bidding is small, the value of the trade is high or the variance in timestamps is low.

\end{proposition}
\begin{proof}
We have
$$\sqrt{2\frac{cgv}{\sqrt{2\pi}\sigma}}-2c\geq\sqrt{\frac{cgv}{\sqrt{2\pi}\sigma}}-c\Leftrightarrow0.414214\sqrt{\frac{cgv}{\sqrt{2\pi}\sigma}}\geq c\Leftrightarrow 0.068447\frac{v}{\sigma}\geq\frac{c}{g}.$$ 
\end{proof}

The previous analysis took the parameter $c$ as given. Next, we consider the revenue when choosing the parameter in an optimal way. The assumption is that the sequencer sets these parameters ex-ante assuming $v$ is distributed according to a CDF $G$ on the non-negative reals.\footnote{The same analysis can be conducted for the optimal choice of $g$. However, choosing $g$ differently has other implications beyond revenue as it influences the time until the finality of a transaction.}
\begin{proposition}
Under time boost, assuming $4\geq gf(0)$, the ex-ante revenue for a shared sequencer is the same as the revenue from two separate sequencers if the parameter $c>0$ is optimally chosen.
\end{proposition}
\begin{proof}
The parameter restriction $4\geq gf(0)$ guarantees that a pure strategy Nash equilibrium exists for any $c>0$ in both cases. Let $G$ be the distribution of values $v$.
For the shared sequencer case the revenue optimisation problem is the following:
$$\max_{c>0} \int_{\tfrac{c}{gf(0)}}^{\infty}(\sqrt{cgf(0)v}-c) dG(v).$$

For the two sequencer case, the optimisation problem is:
$$\max_{c>0} \int_{\tfrac{2c}{gf(0)}}^{\infty}(\sqrt{2cgf(0)v}-2c)dG(v).$$
Now note that with the substitution $\tilde{c}:=2c$ the second problem becomes equivalent to the first problem and therefore has the same objective value
 
\end{proof}

\section{Extensions}
In this section, we consider extensions to the base model presented in the paper so far.

\subsection{More than 2 Players and Chains}
The results can be straightforwardly extended to the case of multiple chains, as long as the distributions of deltas are independent between chains. In that case the probability of winning for player $i$ is given by $F_{\Delta}^n(s_i-s_j)$ and an analogous argument as previously the first order condition
\begin{equation*}
    C'(s^*)=f_{\Delta}(0)F^{n-1}_{\Delta}(0)v=f_{\Delta}(0)v/2^{n-1}.
\end{equation*}
which leads to the equilibrium signal
\begin{equation}\label{equilibrium_separate}
    s^*=(C')^{-1}({f(0)v}/{2^{n-1}}).
\end{equation}
Note that the signal is exponentially decreasing in $n$.

The generalization to more than two players is more complicated, as now the probability of winning does not only depend on the difference of two signaling strategies, but rather on all marginal distributions of signals. However, qualitatively we expect similar results as in the two player case.

\subsection{Smaller cost in case of losing}
We have assumed that in case the trader does not win, he still incurs the full cost of bidding. It is natural to consider the extension where in case of losing, the player pays only an $\alpha$ fraction of the total bid. 
Then, the optimization problem with one sequencer becomes:

\begin{equation}\label{mixed_auction_shared}
    \max_{s_i\geq0}F_{\Delta}(s_i-s_j)(v-C(s_i)) + (1-F_{\Delta}(s_i-s_j))(-\alpha C(s_i)).
\end{equation}

The first summand corresponds to the expected gains of winning the race, and the second summand corresponds to the expected losses of losing the race. 
The first order condition on~\eqref{mixed_auction_shared} with respect to $s$ gives: 

\begin{equation*}
    f_{\Delta}(s_i-s_j)v - (1-\alpha) (f_{\Delta}(s_i-s_j)C(s_i)+F_{\Delta}(s_i-s_j)C'(s_i))-\alpha C'(s_i)=0.
\end{equation*}

Plugging in $s_i=s_j$ and $F_{\Delta}(0)=\frac{1}{2}$ gives an equation from which we can solve for $s_i$:

\begin{equation}\label{eq_cond_mixed_shared}
    f_{\Delta}(0)v - (1-\alpha) (f_{\Delta}(0)C(s_i)+1/2C'(s_i))-\alpha C'(s_i) = 0.
\end{equation}

Note that the case $\alpha=1$ corresponds to the equilibrium condition~\eqref{equilibrium_shared}.
For the constant cost elasticity previously analyzed in Section~\ref{latency}, $C(s)=s^{\beta}$, we can show that the equilibrium solution is decreasing in $\alpha$ provided that the cost elasticity is large enough. In particular,  the signal with positive $\alpha$ is always higher than the solution from~\eqref{equilibrium_shared}:
\begin{proposition}
The equilibrium signal $s$ is decreasing in $\alpha$ for $\beta\geq2$.
\end{proposition}
\begin{proof}
Without loss of generality, assume that $f_{\Delta}(0)=1$ and $v=1$, a normalization.  After plugging in the values for $f_{\Delta}(0)$, $v$ and simplification, the equation~\eqref{eq_cond_mixed_shared} becomes

\begin{equation}\label{indifference_saving_costs}
T(\alpha,s_i) :=1-(1-\alpha)(s_i^{\beta}+1/2\beta s_i^{\beta-1})-\alpha \beta s_i^{\beta-1} = 0. 
\end{equation}

By the implicit function theorem, we get: 

\begin{equation*}
    \frac{\partial T / \partial s_i}{\partial T / \partial \alpha} = \frac{-(1-\alpha)\beta s_i^{\beta-1} - 1/2(1-\alpha)\beta(\beta-1)s_i^{\beta-2}-\alpha\beta(\beta-1)s_i^{\beta-2}}{s_i^{\beta}+1/2\beta s_i^{\beta-1}-\beta s_i^{\beta-1}}. 
\end{equation*}

The nominator of the right-hand side of the above expression is always negative. The denominator: $s_i^{\beta}-\frac{1}{2}\beta s_i^{\beta-1}$ is negative for $\beta$ large enough, since $s_i\leq \beta / 2$ holds for any $\beta\geq 2$. The reason is that if the value is $v=1$, the investment in the equilibrium can not be more than $1$. That is, $s_i$ that solves~\eqref{indifference_saving_costs} equation is decreasing in $\alpha$ parameter. 
\end{proof}
Similarly to the shared sequencer case, we can show that the equilibrium solution for the two separate sequencers is decreasing in $\alpha$ if the cost elasticity is large enough.
The optimization problem in case of 2 sequencers is: 

\begin{align*}
\max_{s_i\geq0}&F_{\Delta}^2(s_i-s_j)(v-2C(s_i)) - 2F_{\Delta}(s_i-s_j)(1-F_{\Delta}(s_i-s_j))(1+\alpha)C(s_i) -
\\
&(1-F_{\Delta}(s_i-s_j))^22\alpha C(s_i).    
\end{align*}

The first term corresponds to the expected gains of winning both races and therefore, obtaining an arbitrage value. The second term corresponds to the expected losses of winning one race and losing the other. In this case, full cost in the first race and fractional cost in the second race are incurred. The third term corresponds to the expected losses of losing both races. In this case, fractional costs in both races are incurred. The first order condition is:

\begin{align*}
    &2F_{\Delta}(s_i-s_J)f_{\Delta}(s_i-s_J)(v-2C(s_i))+F^2_{\Delta}(s_i-s_j)(-2C'(s_i)) - \\
    &(2f_{\Delta}(s_i-s_j)-4F_{\Delta}(s_i-s_j)f_{\Delta}(s_i-s_j))(1+\alpha)C(s_i)-\\
    &2(F_{\Delta}(s_i-s_j)-F^2_{\Delta}(s_i-s_j))(1+\alpha)C'(s_i) - \\
    &2(1-F_{\Delta}(s_i-s_j))(-f_{\Delta}(s_i-s_j))2\alpha C(s_i) - (1-F_{\Delta}(s_i-s_j))^2 2\alpha C'(s_i) = 0,
\end{align*}
which after plugging in $s_i=s_j$ and $F(0)=\frac{1}{2}$ becomes

\begin{align}\label{separate_saving}
    f_{\Delta}(0) (v-2C(s_i))-(1+\alpha)C'(s_i) + 2\alpha f_{\Delta}(0)C(s_i) = 0.
\end{align}
Note that the case $\alpha=1$ corresponds to the equilibrium condition~\eqref{condition_separate}. A similar calculation as before establishes the monotonicity of the equilibrium signal.
Comparing the equilibrium cost between shared and separate sequencing in this extended model is less tractable than in the baseline model. However, we can obtain numerical comparison results by solving equations~\eqref{eq_cond_mixed_shared} and~\eqref{separate_saving}, for given $C$, $F$ and $\alpha$.

\subsection*{Other Cost Functions, Noise Models and Empirical Work}
We currently explore other extensions of the model as well. A common design proposal for transaction ordering is {\it frequent ordering auction} design where transactions are processed in batches according to bids attached to the transactions. Randomness comes into play because whether or not a transaction makes it into the batch is determined by time stamps. Towards the end of a batch, this introduces randomness about which transactions make it into the current or the next batch. An ordering auction formulation would require a different cost function specification, as well as a non-additive noise model. Another extension of our model that we currently explore is to calibrate the parameters of our model according to real-world data. This gives a precise sense of the magnitude of the effects studied.

\bibliographystyle{plain}
\bibliography{refs}

\begin{thebibliography}{1}

\bibitem{mev}
Philip Daian, Steven Goldfeder, Tyler Kell, Yunqi Li, Xueyuan Zhao, Iddo
  Bentov, Lorenz Breidenbach, and Ari Juels.
\newblock {Flash Boys 2.0: Frontrunning in Decentralized Exchanges, Miner
  Extractable Value, and Consensus Instability}.
\newblock In {\em 2020 {IEEE} Symposium on Security and Privacy, {SP} 2020, San
  Francisco, CA, USA, May 18-21, 2020}, pages 910--927. {IEEE}, 2020.

\bibitem{arbitrum}
Harry~A. Kalodner, Steven Goldfeder, Xiaoqi Chen, S.~Matthew Weinberg, and
  Edward~W. Felten.
\newblock Arbitrum: Scalable, private smart contracts.
\newblock In William Enck and Adrienne~Porter Felt, editors, {\em 27th {USENIX}
  Security Symposium, {USENIX} Security 2018, Baltimore, MD, USA, August 15-17,
  2018}, pages 1353--1370. {USENIX} Association, 2018.

\bibitem{timeboost}
Akaki Mamageishvili, Mahimna Kelkar, Jan~Christoph Schlegel, and Edward~W.
  Felten.
\newblock Buying time: Latency racing vs. bidding in fair transaction ordering.
\newblock {\em CoRR}, abs/2306.02179, 2023.

\bibitem{cross_domain_sok}
Conor McMenamin.
\newblock {SoK: Cross-Domain MEV}.
\newblock {\em arXiv preprint arXiv:2308.04159}, 2023.

\bibitem{cross_domain}
Alexandre Obadia, Alejo Salles, Lakshman Sankar, Tarun Chitra, Vaibhav
  Chellani, and Philip Daian.
\newblock Unity is strength: A formalization of cross-domain maximal
  extractable value.
\newblock {\em arXiv preprint arXiv:2112.01472}, 2021.

\bibitem{10174971}
Johan~Hagelskjar Sjursen, Weizhi Meng, and Wei-Yang Chiu.
\newblock A closer look at cross-domain maximal extractable value for
  blockchain decentralisation.
\newblock In {\em 2023 IEEE International Conference on Blockchain and
  Cryptocurrency (ICBC)}, pages 1--3, 2023.

\bibitem{espresso}
Espresso Systems.
\newblock {The Espresso Sequencer}, 2023.

\end{thebibliography}

\end{document}